\documentclass[letterpaper,11pt]{article}
\usepackage{tabularx} 
\usepackage{amsmath}  
\usepackage{amsthm}
\usepackage{amsfonts}
\usepackage{tikz}
\usetikzlibrary{fit,graphs,graphs.standard,quotes,shapes,arrows,automata}
\usepackage{graphicx} 
\usepackage[margin=1.3in,letterpaper]{geometry} 
\usepackage{cite} 
\usepackage[final]{hyperref} 
\hypersetup{
	colorlinks=true,       
	linkcolor=red,        
	citecolor=blue,        
	filecolor=magenta,     
	urlcolor=blue         
}
\newtheorem{theorem}{Theorem}

\newtheorem{question}{Question}
\newtheorem{corollary}{Corollary}
\newtheorem{proposition}{Proposition}

\newcommand{\comment}[1]{}

\DeclareMathOperator{\bin}{bin}

\begin{document}
	
\title{Subset Synchronization in Monotonic Automata}
\author{Andrew Ryzhikov$^{1, 2}$, Anton Shemyakov$^3$}
\date{%
	$^1$ Universit\'e Grenoble Alpes, Laboratoire G-SCOP, 38031 Grenoble, France\\%
	$^2$ United  Institute of Informatics Problems  of NASB, 
	220012 Minsk, Belarus\\
	$^3$ Belarusian State University, 220030 Minsk, Belarus\\
	{\tt ryzhikov.andrew@gmail.com, shemyanton@gmail.com}
}
\maketitle

\begin{abstract} We study extremal and algorithmic questions of subset and careful synchronization in monotonic automata. We show that several synchronization problems that are hard in general automata can be solved in polynomial time in monotonic automata, even without knowing a linear order of the states preserved by the transitions. We provide asymptotically tight bounds on the maximum length of a shortest word synchronizing a subset of states in a monotonic automaton and a shortest word carefully synchronizing a partial monotonic automaton. We provide a complexity framework for dealing with problems for monotonic weakly acyclic automata over a three-letter alphabet, and use it to prove NP-completeness and inapproximability of problems such as {\sc Finite Automata Intersection} and the problem of computing the rank of a subset of states in this class. We also show that checking whether a monotonic partial automaton over a four-letter alphabet is carefully synchronizing is NP-hard. Finally, we give a simple necessary and sufficient condition when a strongly connected digraph with a selected subset of vertices can be transformed into a deterministic automaton where the corresponding subset of states is synchronizing.
\end{abstract}
	
\section{Introduction}

Let $A = (Q, \Sigma, \delta)$ be a deterministic finite automaton (which we further simply call an {\em automaton}), where $Q$ is the set of its states, $\Sigma$ is a finite alphabet and $\delta: Q \times \Sigma \to Q$ is a transition function. Note that our definition of automata does not include initial and accepting states. The mapping $\delta$ can be inductively extended to the mapping $Q \times \Sigma^* \to Q$, which we also denote as $\delta$: for each word $xw$, where $x$ is a letter, take $\delta(q, xw) = \delta(\delta(q, x), w)$. An automaton is called {\it synchronizing} if there exists a word that maps every its state to some fixed state. Such word is also called {\em synchronizing}. Synchronizing automata play an important role in manufacturing, coding theory and biocomputing, and model systems that can be controlled without knowing their actual state \cite{Volkov2008}.

Synchronizing automata model devices that can be reset, by applying a synchronizing word, to some particular state without having any information about their current state. Automata with a synchronizing set of states model devices that can be reset to a particular state with some partial information about the current state, namely when it is known that the current state belongs to a synchronizing subset of states. A set $S \subseteq Q$ of states of an automaton $A = (Q, \Sigma, \delta)$ is called {\em synchronizing} if there exists a word $w \in \Sigma^*$ and a state $q \in Q$ such that the word $w$ maps each state $s \in S$ to the state $q$. The word $w$ is said to {\em synchronize} the set $S$. It follows from the definition that an automaton is synchronizing if and only if the set $Q$ of all its states is synchronizing.

In this paper, we deal with monotonic and weakly acyclic automata. An automaton $A = (Q, \Sigma, \delta)$ is called {\em monotonic} if there is a linear order $\le$ of its states such that for each $x \in \Sigma$ if $q_1 \le q_2$ then $\delta (q_1, x) \le \delta(q_2,x)$. In this case we say that the transitions of the automaton {\em preserve}, or {\em respect} this order. Monotonic automata play an important role in the part-orienting process in manufacturing \cite{Ananichev2004}, some connections of monotonic automata with infinite games are described in \cite{Kopczynski2006} and \cite{KopczynskiThesis2008}. Once the order $q_1, \ldots, q_n$ of the states is fixed, we denote $[q_i, q_j] = \{q_\ell \mid i \le \ell \le j\}$, and $\min S$, $\max S$ as the minimum and maximum states of $S \subseteq Q$ with respect to the order. The following open problem is mentioned in \cite{Shcherbak2006}, showing that monotonic automata are not fully understood, and require more investigation.

\begin{question}
	Find a combinatorial characterization (for example, using regular expressions) of languages recognized by monotonic automata.
\end{question}

An automaton $A = (Q, \Sigma, \delta)$ is called {\em weakly acyclic} if there exists an order of its states $q_1, \ldots, q_n$ such that if $\delta(q_i, x) = q_j$ for some $x \in \Sigma$, then $i \le j$. Note that a monotonic automaton does not have to be weakly acyclic, and vice versa. Both weakly acyclic and monotonic automata present proper subclasses of a widely studied class of aperiodic automata \cite{Volkov2008}. An automaton is called {\em orientable}, if there exists a cyclic order of its states that is preserved by all transitions of the automaton (see \cite{Volkov2008} for the discussion of this definition). Each monotonic automaton is obviously orientable. An automaton is called {\em strongly connected} if any its state can be mapped to any other state by some word.

The two fundamental directions is studying synchronization of automata are extremal (bounding the length of a shortest synchronizing word) and algorithmic (exploring the complexity of deciding synchronizability and finding a shortest synchronizing word) questions. 

From the extremal point of view, it is known that any synchronizing $n$-state automaton can be synchronized by a word of length at most $\frac{n^3 - n}{6}$ \cite{Pin1983}, and the famous \v{C}ern{\'y} conjecture states that the length of such word is at most $(n - 1)^2$ \cite{Volkov2008}. A slightly better but still cubic bound is reported in \cite{Szykula2017}. For words synchronizing a subset of states, the situation is quite different. It is known that the length of a shortest word synchronizing a subset of states in a binary strongly connected automaton can be exponential in the number of states of the automaton \cite{Vorel2016}. In weakly acyclic automata, there is a quadratic upper bound on the length of such words \cite{Ryzhikov2017}. For orientable $n$-state automata, a tight $(n - 1)^2$ upper bound on the length of a shortest word synchronizing a subset of states is known \cite{Eppstein1990}.

Checking whether an automaton is synchronizing can be performed in polynomial time \cite{Volkov2008}, but checking whether a given subset of states in an automaton is synchronizing (the {\sc Sync Set} problem) is a PSPACE-complete problem in binary strongly connected automata \cite{Vorel2016}, and a NP-complete problem in binary weakly acyclic automata \cite{Ryzhikov2017}.

Eppstein \cite{Eppstein1990} provides a polynomial algorithm for the {\sc Sync Set} problem, as well as for some other problems, in orientable automata. However, the proposed algorithms assume that a cyclic order of the states preserved by the transitions is known. Since the problems of recognizing monotonic and orientable automata are NP-complete \cite{Szykula2015}, a linear or cyclic order preserved by the transitions of an automaton cannot be computed in polynomial time unless P = NP. Thus, we should avoid using these orders explicitly in algorithms, so we have to investigate other structural properties of monotonic automata. As  shown in this paper, several synchronization problems are still solvable in polynomial time in monotonic automata without knowing an order of states preserved by the transitions.

Approximating the length of a shortest word synchronizing a $n$-state automaton within a factor of $O(n^{1 - \epsilon})$ for any $\epsilon > 0$ in polynomial time is impossible unless P = NP \cite{Gawrychowski2015}. For finding the length of a shortest word synchronizing a subset of states in binary weakly acyclic automata a similar inapproximability bound holds \cite{Ryzhikov2017}.

A problem closely connected to subset synchronization is careful synchronization of partial automata. A {\em partial automaton} $A$ is a triple $(Q, \Sigma, \delta)$, where $Q$ and $\Sigma$ are the same as in the definition of a finite deterministic automaton, and $\delta$ is a partial transition function (i.e., a transition function which may be undefined for some argument values). A word $w$ is said to {\em carefully synchronize} a partial automaton $A$ if it maps all its states to the same state, and each mapping corresponding to a prefix of $w$ is defined for each state. The automaton $A$ is then called {\em carefully synchronizing}.

The length of a shortest word carefully synchronizing a $n$-state partial automaton is also a subject of research. Rystsov \cite{Rystsov1980}, Martyugin \cite{Martyugin2010lower}, Vorel \cite{Vorel2016} and de Bondt et al. \cite{Bondt2017} propose consecutive improvements of (exponential) lower bounds for this value, both in the case of constant and non-constant alphabets. Rystsov \cite{Rystsov1980} provides an upper bound of $O(3^{\frac{n}{3}})$ on this value. A simple relation between careful synchronization and subset synchronization is provided by Lemma 1 of \cite{Vorel2016}.

Deciding whether a partial automaton is carefully synchronizing is PSPACE-complete for binary partial automata \cite{Martyugin2010}, and moreover for binary strongly connected partial automata \cite{Vorel2016}. It is also NP-hard for aperiodic partial automata over a three-letter alphabet \cite{Ryzhikov2017}.

A synchronizing set of states can be considered as a set compressible to one element. A more general case of a set compressible to a set of size $r$ is defined by the notion of the rank of a subset. Given an automaton $A = (Q, \Sigma, \delta)$, the {\em rank} of a word $w \in \Sigma^*$ with respect to a set $S \subseteq Q$ is the size of the image of $S$ under the mapping defined by $w$ in $A$, i.e., the number $|\{\delta(s, w) \mid s \in S\}|$. The {\em rank} of an automaton (respectively, of a subset of states) is the minimum among the ranks of all words $w \in \Sigma^*$ with respect to the whole set $Q$ of states of the automaton (respectively, to the subset of states). It follows from the definition that a set of states has rank $1$ if and only if it is synchronizing. A state in an automaton is a {\em sink state} if all letters map this state to itself.

For $n$-state monotonic automata of rank at most $r$, Ananichev and Volkov \cite{Ananichev2004} show an upper bound of $n - r$ on the length of a shortest word of rank at most $r$, and also provide bounds on the length of a shortest word of interval rank at most $r$. Shcherbak \cite{Shcherbak2006} continues the investigation of words of bounded interval rank in monotonic automata. Ananichev \cite{Ananichev2005} provides bounds on the length of a shortest word of rank $0$ in partial monotonic automata of rank $0$ (a partial automaton is called {\em monotonic} if there exists a linear order of its states preserved by all defined transitions).

In this paper, we study both extremal and algorithmic questions of subset synchronization in monotonic automata. In Section \ref{sec-struct}, we provide structural results about synchronizing sets of states in monotonic automata and give algorithmic consequences of this results.
In Section \ref{sec-bound}, we provide lower and upper bounds on the maximum length of shortest words synchronizing a subset of states in monotonic automata, and show some lower bounds for related problems. In Section \ref{sec-complex}, we provide NP-hardness and inapproximability of several problems related to subset synchronization and careful synchronization of monotonic automata. In Section~\ref{sec-road} we give necessary and sufficient conditions when a strongly connected digraph can be colored resulting in an automaton with a pre-defined synchronizing set.

A conference version of this paper was published in \cite{Ryzhikov2017Monotonic}. Besides presenting new results, this paper corrects errors of the conference version.

\section{Structure of Synchronizing Sets} \label{sec-struct}

Let $A$ be an automaton, and $S$ be a subset of its states. In general, if any two states in $S$ can be synchronized (i.e., form a synchronizing set), $S$ does not necessarily have to be synchronizing, as it is shown by the following theorem.

\begin{theorem} \label{tm-pairwise-no}
	For any positive integer $k_0$, there exists a binary weakly acyclic automaton $A$ and a subset $S$ of its states such that $|S| \ge k_0$, each pair of states in $S$ in synchronizing, but the rank of $S$ equals $|S| - 1$.
\end{theorem}
\begin{proof}
	Consider the following automaton $A = (Q, \{0, 1\}, \delta)$. Let $S = \{s_0, \ldots, s_{k - 1}\}$. Let $k = 2^\ell$ for an integer number $\ell$, and let $\bin(i)$ be a word which is equal to the binary representation of $i$ of length $\ell$ (possibly with zeros at the beginning). We introduce new states $t_i, p_i$ for $0 \le i \le k - 1$, a state $f$, and new intermediate states in $Q$ as follows. For each $s_i$, $0 \le i \le k - 1$, consider a construction sending $s_i$ to $f$ for a word $\bin(i)$, and to $t_i$ by any other word of length $\ell$. 
	
	
	For each $t_i$, consider the same construction sending $t_i$ to $f$ for a word $\bin(i)$, and to $p_i$ otherwise. For each $i$, define both transitions from $p_i$ as self-loops. Define both transitions from $f$ as self-loops.
	
	In this construction, each word applied after a word of length $2\ell$ obviously has no effect. Consider a word $w$ of length $2\ell$, $w = w_1w_2$, where both $w_1$ and $w_2$ have length $\ell$. If $w_1 = w_2$, then the image of $S$ under the mapping defined by $w$ has size $k$. Otherwise, $w$ synchronizes two states $s_i$ and $s_j$ with $\bin(i) = w_1$ and $\bin(j) = w_2$ and maps all other states to different states. Thus, the rank of $S$ equals $k - 1$.
\end{proof}

The size of the whole automaton is $O(|S| \log |S|)$, thus $S$ can be large comparing to the size of the whole set of states in the automaton.

Since the {\sc Sync Set} problem is PSPACE-complete in strongly connected automata, pairwise synchronization of states in a subset does not imply that this subset is synchronizing for this class of automata unless P = PSPACE. Thus, it is reasonable to ask the following question.

\begin{question}
	How large can be the rank of a subset of states in a strongly connected automaton such that each pair of states in this subset can be synchronized?
\end{question}

For the rest of section, fix a monotonic automaton $A = (Q, \Sigma, \delta)$ and an order $q_1, \ldots, q_n$ of its states preserved by all transitions. As shown by the next theorem, the situation in monotonic automata is in some sense opposite to the situation described by Theorem \ref{tm-pairwise-no}.
 
\begin{theorem}\label{tm-pairwise}
	Let $S \subseteq Q$ be a subset of states of $A$. Then $S$ is synchronizing if and only if any two states in $S$ can be synchronized.
\end{theorem}
\begin{proof}
	Obviously, any subset of a synchronizing set is synchronizing.
	
	In the other direction, if any two states in $S$ can be synchronized, then the minimal state $q_\ell = \min S$ and the maximal state $q_r = \max S$ in $S$ can be synchronized by a word $w \in \Sigma^*$. Let $q = \delta(q_\ell, w) = \delta(q_r, w)$. Then the interval $[q_\ell, q_r] = \{q_\ell, \ldots, q_r\}$ is synchronized by $w$, because each state of $[q_\ell, q_r]$ is mapped to the interval $[\delta(q_\ell, w), \delta(q_r, w)]$ = $\{q\}$, since $A$ is monotonic. Thus, $S \subseteq [q_\ell, q_r]$ is synchronizing.
\end{proof}

\begin{corollary} \label{cor-check-set}
	The problem of checking whether a given set $S$ is synchronizing can be solved in $O(|Q|^2 \cdot |\Sigma|)$ time and space for monotonic automata.
\end{corollary}
\begin{proof}
	By Theorem \ref{tm-pairwise} it is enough to check that each pair of states in $S$ can be synchronized, which can be done by solving the reachability problem in the subautomaton of the power automaton, built on all $2$-element and $1$-element subsets of $Q$ \cite{Volkov2008}. There are $\frac{|Q|(|Q| + 1)}{2}$ states in this subautomaton $A^2$. We need to check that from each state $\{q_i, q_j\}$, $q_i, q_j \in S$, in $A^2$ some singleton set is reachable. Consider the underlying digraph of $A^2$ and reverse all arcs in it. Then we need to check that in this new digraph each vertex $\{q_i, q_j\}$, $q_i, q_j \in S$, is reachable from some singleton. To check it, run breadth-first search simultaneously from all singletons \cite{Cormen2009}.
	
	To construct the subautomaton we need $O(|Q|^2 \cdot |\Sigma|)$ time and space, and breadth-first search requires time and space linear in the number of arcs of the digraph.
\end{proof}

\begin{corollary}
	A shortest word synchronizing a given subset $S$ of states can be found in $O(|Q|^4 \cdot |\Sigma|)$ time and $O(|Q|^2 \cdot |\Sigma|)$ space for monotonic automata.
\end{corollary}
\begin{proof}
	Consider the following algorithm. For each pair of states, find a shortest word synchronizing this pair. This can be done by solving the shortest path problem in the subautomaton of the power automaton, build on all $2$-element and $1$-element subsets of $Q$ \cite{Volkov2008}. Let $W$ be the set of all such words that synchronize $S$. Output the shortest word in $W$.
	
	By an argument similar to the proof of Theorem \ref{tm-pairwise}, any shortest word synchronizing $\{\min S, \max S\}$ is a shortest word synchronizing $S$, thus the algorithm finds a shortest word synchronizing $S$. Since finding a shortest synchronizing word for a pair of states requires $O(|Q|^2 \cdot |\Sigma|)$ time and there are $O(|Q|^2)$ such words, the set $W$ can be found in $O(|Q|^4 \cdot |\Sigma|)$ time. Finding a shortest word in $W$ synchronizing $S$ requres additional $O(|Q|^4)$ time, since we have to check each word. Thus, the total time required by the algorithm is $O(|Q|^4 \cdot |\Sigma|)$.
	
	The algorithm requires $O(|Q|^2 \cdot |\Sigma|)$ space for the subautomaton construction. We don't have to store $W$, since we can check the words in $W$ one by one and store only the shortest one, so we need $O(|Q|^2 \cdot |\Sigma|)$ space.
\end{proof}

\comment{Define the {\em synchronization graph} $G(A)$ of an automaton $A = (Q, \Sigma, \delta)$ as follows. The set of vertices of $G(A)$ is $Q$. Two vertices are adjacent in $G$ if and only if the two corresponding states can be synchronized. 

The following fact follows from the proof of Theorem 4 in \cite{Ryzhikov2017} and shows that in general the structure of synchronization graphs can be very complicated.

\begin{proposition}
	Each graph is an induced subgraph of a synchronization graph of some binary weakly acyclic automaton.
\end{proposition}

However, synchronization graphs of monotonic automata are very special. An {\em interval graph} is the intersection graph of a set of intervals of a line.

\begin{theorem}\label{tm-interval}
	Synchronization graphs of monotonic automata are interval graphs.
\end{theorem} 
\begin{proof}
	For each state $q \in Q$ in $A$, consider the set $S_q$ of all such states $q'$ that $q$ and $q'$ can be synchronized. It follows from the proof of Theorem \ref{tm-pairwise} that each maximal (by inclusion) synchronizing set is an interval $[q_i, q_j]$ for some $q_i, q_j$. Thus, each set $S_q$ forms an interval, which is the union of all intervals that are maximal synchronizing sets containing $q$.
	
	Observe that two states can be synchronized if and only if the corresponding intervals intersect \textcolor{red}{?????????}. Thus, we obtain that the synchronization graph of $A$ is the intersection graph of the constructed set of intervals. 
\end{proof}

Theorem \ref{tm-interval} implies the following.
}

\begin{corollary}
	A synchronizing subset of states of maximum size can be found in $O(|Q|^4 \cdot |\Sigma| + |Q|^5)$ time and $O(|Q|^2 \cdot |\Sigma|)$ space in monotonic automata.
\end{corollary}
\begin{proof}
	For each synchronizing pair $q_i, q_j$ of states, find a word synchronizing this pair (in the same way as described in Corollary \ref{cor-check-set}), and the find all states that are mapped by this word to the same state as  $q_i$ and $q_j$. Output the pair with the largest synchronizing set constructed in such a way.
	
	To prove that this algorithm is correct, observe that each word synchronizing a pair $q_i, q_j$ of states synchronizes also all the  states $q_k$, $q_i < q_k < q_j$. On the other hand, if $q_i = \min S$ and $q_j = \max S$ for a synchronizing set $S$ of maximum size (which is an interval), any word synchronizing $q_i$ and $q_j$ synchronizes only $S$.
	
	The described algorithm requires $O(|Q|^2 \cdot (|Q|^2 \cdot |\Sigma| + |Q|^3))$ (for each pair of states, we need to find a synchronizing word which requires $O(|Q|^2 \cdot |\Sigma|)$ time, and then apply this word to each state, which requires $O(|Q|^3))$ time. Since we need to store only the set of maximum size, the algorithm requires $O(|Q|^2 \cdot |\Sigma|)$ space.
\end{proof}

The problem of finding a synchronizing subset of states of maximum size in general automata is PSPACE-complete \cite{Ryzhikov2017}. T\"{u}rker and Yenig\"{u}n \cite{Turker2015} study a variation of this problem, which is to find a set of states of maximum size that can be mapped by some word to a subset of a given set of states in a given monotonic automaton. They reduce the {\sc N-Queens Puzzle} problem \cite{Bell2009} to this problem to prove its NP-hardness. However, their proof is unclear, since in the presented reduction  the input has size $O(\log N)$, and the output size is polynomial in $N$.

The algorithms proposed in this section run polynomial time, but the degrees of these polynomials are quite high. A natural question is to find faster algorithms for the described problems. Another interesting quiestion is whether the results can be generalized to oriented automata.

\section{Lower Bounds for Synchronizing Words}\label{sec-bound}

The length of a shortest word synchronizing a $n$-state monotonic automaton is at most~$n - 1$ \cite{Ananichev2004}. In this section we investigate a more general question of bounding the length of a shortest word synchronizing a subset of states in a $n$-state synchronizing automaton. For a more general class of oriented automata a bound of $(n - 2)^2$ is known \cite{Eppstein1990}, but for monotonic automata a smaller upper bound can be proved.

\begin{theorem}\label{tm-sycn-upper}
	Let $S$ be a synchronizing set of states in a monotonic $n$-state automaton~$A$. Then for $n \ge 8$ the length of a shortest word synchronizing $S$ is at most $\frac{(n - 2)^2}{4}$.
\end{theorem}
\begin{proof}
	Let $A = (Q, \Sigma, \delta)$, and $\{q_1, \ldots, q_n\}$ be an order of the states preserved by all transitions of $A$. Define $q_\ell = \min S$, $q_r = \max S$. We can assume that $S$ can be mapped only to the states in $[q_\ell, q_r]$. Indeed, assume without loss of generality that $q_i$, $i < \ell$, is a state such that $S$ can be mapped to $q_i$, and it is the smallest such state. Then by monotonicity there exists a word mapping $q_r$ to $q_i$ by taking only transitions going to smaller states. This word then synchronizes $S$ and has length at most $n - 1 \le \frac{(n - 2)^2}{4}$ for $n \ge 8$.
	
	
	
	Now we can assume that $S$ can be mapped only to states in $[q_\ell, q_r]$. This means that $S$ can be mapped to the states $q_i, q_j$ in $[q_\ell, q_r]$, and no state outside $[q_i, q_j]$ is reachable from any state of $[q_i, q_j]$. Indeed, the set of states reachable from both $q_\ell, q_r$ contains $q_i$ and $q_j$, and if some state outside $[q_i, q_j]$ is reachable from $[q_i, q_j]$ we can synchronize $S$ to a state outside $[q_i, q_j]$, which contradicts the definition of the interval $[q_i, q_j]$. If both $q_\ell$ and $q_r$ are mapped to states inside $[q_i, q_j]$, they can be then synchronized by applying a word of length at most $j - i$, for example by applying a word $w'$ composed of only letters mapping the consecutive images of $q_j$ to states with smaller indexes by the same reasoning as below. By our assumptions, $q_i$ is reachable from each state in $[q_i, q_j]$, thus such a word $w'$ exists.
 	
 	Suppose now that $w = w_1 \ldots w_m$ is a shortest word synchronizing $S$. Consider the sequence of pairs $(t_k, s_k) = (\delta(q_\ell, w_1 \ldots w_k), \delta(q_r, w_1 \ldots w_k))$, $k = 1, 2, \ldots, m$. As $w$ is a shortest word synchronizing $S$, and synchronization of $S$ is equivalent to synchronization of $\{q_\ell, q_r\}$, no pair appears in this sequence twice, and the only pair with equal components is $(s_m, t_m)$. Because of monotonicity, $t_k \le s_k$ for each $1 \le k \le m$. Thus, the maximum length of $w$ is reached when $q_i = q_j$, since after both images are in $[q_i, q_j]$ the remaining length of a synchronizing word is at most $j - i$. Observe that if $q_1$ or $q_n$ is in $S$, $S$ again can be synchronized by a word of length $n - 1$. Thus, we can assume that $|S| \le n - 2$, and thus the length of $w$ is at most $(i - 1)(n - 3 - i) \le \frac{(n - 2)^2}{4}$.	
\end{proof}

	The bound is almost tight for monotonic automata over a three-letter alphabet as shown by the following example.
	
\begin{theorem} \label{tm-sync-lower3}
	For each $m \ge 1$, there exist a $(2m + 3)$-state monotonic automaton $A$ over a three-letter alphabet, which has a subset $S$ of states, such that the length of a shortest word synchronizing $S$ is $m^2 + m$.
\end{theorem}
\begin{proof}
	Consider the following monotonic automaton $A = (Q, \Sigma, \delta)$, $Q = \{q_1, \ldots, q_{2m + 3}\}$. Let $\Sigma = \{0, 1, 2\}$. Let states $q_1$, $q_{m+2}$ and $q_{2m+3}$ be sink states. For every state $q_i$, $2 \le i \le m + 1$, we set $\delta(q_i, 0) = q_{i + 1}$, $\delta(q_i, 1) = q_i$, $\delta(q_i, 2) = q_1$. For every state $q_i$, $m + 4 \le i \le 2m + 2$, we set $\delta(q_i, 0) = q_{2m + 3}$, $\delta(q_i, 1) = q_{i - 1}$, $\delta(q_i, 2) = q_i$. Finally we define $\delta(q_{m+3},0) = q_{2m+2}$, $\delta(q_{m + 3}, 1) = q_{m+3}$, $\delta(q_{m + 3}, 2)=q_{m + 2}$.
	See Figure \ref{fig-subset-3} for an illustration of the construction.
	
	\setlength{\unitlength}{2pt}
	\begin{figure}[hbt]
		\begin{center}
			\begin{tikzpicture}[->,>=latex',
			vertex/.style={circle, draw=black, fill, minimum width=1.5mm, inner sep=0pt, outer sep=0pt},
			every label/.style={inner sep=0pt, minimum width=0pt, label distance=0.1mm},
			yscale=1,
			xscale=1.5
			]
			\graph[nodes=vertex, empty nodes, no placement] {
				{
					q2[x=0,y=0,label=below:$q_2$] -> [edge]
					q3[x=1,y=0,label=below:$q_3$] -> [edge label=$\cdots$, swap]
					qm1[x=2,y=0,label=below:$q_{m+1}$] -> [edge]
					qm2[x=3,y=0,label=below:$q_{m+2}$]
				};
				{ 	
					q2m2[x=7,y=0,label=below:$q_{2m+2}$] -> [edge label=$\ldots$, dashed]
					qm5[x=6,y=0,label=below:$q_{m+5}$] -> [edge, dashed]
					qm4[x=5,y=0,label=below:$q_{m+4}$] -> [edge, dashed]
					qm3[x=4,y=0,label=below:$q_{m+3}$] -> [edge, dotted]
					qm2
				};
				{
					q1[x=-1,y=0,label=below:$q_1$]
				};
				
				q2 -> [edge, bend right = 40, dotted] q1;
				q3 -> [edge, bend right = 60, dotted] q1;
				qm1 -> [edge, bend right = 80, dotted] q1;
				
				q2m3[x=8,y=0,label=below:$q_{2m+3}$] -> [edge, loop above] q2m3;
				
				qm3 -> [edge, bend left = 50] q2m2;
				qm4 -> [edge, bend left = 50] q2m3;
				qm5 -> [edge, bend left = 50] q2m3;
				
				q1 -> [edge, loop above] q1;
				qm2 -> [edge, loop above] qm2;
				q2m2 -> [edge, bend left = 50] q2m3;
				
			};
			\end{tikzpicture}
			\caption{The automaton providing a lower bound for subset synchronization in monotonic automata over a three-letter alphabet. Solid arrows represent transitions for the letter $0$, dashed -- for the letter $1$, dotted -- for the letter $2$. The states $q_1, q_{m + 2}, q_{2m + 3}$ are sink states, self-loops are omitted.} \label{fig-subset-3}
		\end{center}
	\end{figure}
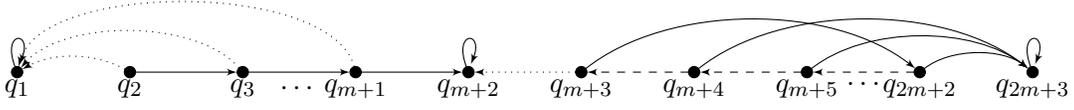
	
	All transitions of $A$ respect the order $q_1, \ldots, q_{2m + 3}$, so $A$ is monotonic. Let us show that the shortest word synchronizing the set $S = \{q_2, q_{2m+2}\}$ is $w = 1^{m-1}(01^{m-1})^m 2$. Let $S'$ be a set of states such that $q_i, q_j \in S'$, $2 \le i \le m + 1$, $m + 3 \le j \le 2m + 2$. The set $S'$ can be mapped only to $q_{m + 2}$, because $A$ is monotonic. Hence if any state of $S'$ is mapped by a word to $q_1$ or to $q_{2m + 3}$, then this word cannot synchronize $S'$.
	
	We start with the set $S = \{q_2, q_{2m+2}\}$. There is only one letter $1$ that does not map the state $q_2$ to $q_1$ or the state $q_{2m+2}$ to $q_{2m+3}$, and maps $S$ not to itself. Indeed, $0$ maps $q_{2m+2}$ to $q_{2m+3}$ and $2$ maps $q_2$ to $q_1$. Thus, any shortest synchronizing word can start only with $1$. Consider now the set $\{\delta(q_2, 1), \delta(q_{2m+2}, 1)\} = \{q_2, q_{2m+1}\}$. There is only letter $1$ that does not map the state $q_3$ to $q_1$, or $q_{2m + 1}$ to $q_{2m+3}$ and maps this set not to itself. Indeed, $0$ maps $q_{2m+1}$ to $q_{2m+3}$ and $2$ maps $q_3$ to $q_1$. So the second letter of the shortest synchronizing word can only be $1$. By a similar reasoning (at each step there is exactly one letter that maps a pair of states not to itself and does not map the states to the sink states $q_1$ and $q_{2m + 3}$), we deduce that any shortest synchronizing word has to begin with $1^{m-1}(01^{m-1})^m$ and it is easy too see that $1^{m-1}(01^{m-1})^m 2$ synchronizes $S$. Thus, $w$ is a shortest word synchronizing $S$, and its length is $m^2 + m$.
\end{proof}

For a $n$-state automaton, the lower bound on the length of a shortest word in this theorem is $\frac{(n - 2)^2 - 1}{4}$, which is very close to the lower bound $\frac{(n - 2)^2}{4}$ from Theorem \ref{tm-sycn-upper}.

By taking $q_2$ and $q_{2m+2}$ as initial states in two equal copies of the automaton in the proof of Theorem \ref{tm-sync-lower3}, and taking $q_{m + 2}$ as the only accepting state in both copies, we obtain the following result.

\begin{corollary}
	A shortest word accepted by two $(2m + 3)$-state monotonic automata which differ only by their initial states can have length $m^2 + m$.
\end{corollary}

For binary monotonic automata, our lower bound is slightly smaller, but still quadratic.

\begin{theorem}
	For each $m \ge 1$, there exist a $(4m + 3)$-state binary monotonic automaton~$A$, which has a subset $S$ of states such that the length of a shortest word synchronizing~$S$ is at least $m^2$.
\end{theorem}
\begin{proof} 
	Consider the following automaton $A = (Q, \Sigma, \delta)$ with $Q = \{q_1, \ldots, q_{4m + 3}\}$, $\Sigma = \{0,1\}$. Define $\delta$ as follows. Set $q_1, q_{2m + 2}, q_{4m + 3}$ to be sink states. Define $\delta(q_i, 1) = q_{i - 1}$ for all $i \not = 1, 2m + 2, 4m + 3$. For each $i$, $2 \le i \le m + 1$, define $\delta(q_i, 0) = q_{i + m}$, and for each $i$, $m + 2 \le i \le 2m + 1$, define $\delta(q_i, 0) = q_{2m + 2}$. For each $i$, $2m + 3 \le i \le 3m + 3$, define $\delta(q_i, 0) = q_{m + i - 1}$, and for each $i$,  $3m + 4 \le i\le 4m + 2$, define $\delta(q_i, 0) = q_{4m + 3}$. The defined binary automaton is monotonic, since all its transitions respect the order $q_1, \ldots, q_{4m + 3}$. See Figure \ref{fig-subset-2} for an example of the construction.
	
	\setlength{\unitlength}{2pt}
	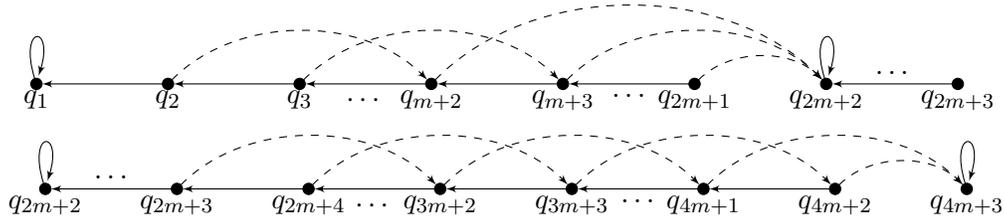
\begin{figure}[hbt]
		\begin{center}
			\begin{tikzpicture}[->,>=latex',
			vertex/.style={circle, draw=black, fill, minimum width=1.5mm, inner sep=0pt, outer sep=0pt},
			every label/.style={inner sep=0pt, minimum width=0pt, label distance=0.1mm},
			yscale=1.5,
			xscale=1.75
			]
			\graph[nodes=vertex, empty nodes, no placement] {

				q1[x=-2,y=0,label=below:$q_1$] -> [edge, loop above] q1;
				
				q2m2[x=4,y=0,label=below:$q_{2m + 2}$];
				
				q2m1[x=3,y=0,label=below:$q_{2m + 1}$] -> [edge label=$\ldots$]
				qm3[x=2,y=0,label=below:$q_{m + 3}$] ->
				qm2[x=1,y=0,label=below:$q_{m + 2}$] -> [edge label=$\cdots$]
				q3[x=0,y=0,label=below:$q_3$] ->
				q2[x=-1,y=0,label=below:$q_2$] -> q1;
				
				q2m3[x=5,y=0,label=below:$q_{2m + 3}$] -> [edge label={$\ldots$}, swap] q2m2;

				q2 -> [edge, bend left = 50, dashed] qm2;
				q3 -> [edge, bend left = 50, dashed] qm3;
				qm2 -> [edge, bend left = 50, dashed] q2m2;
				qm3 -> [edge, bend left = 50, dashed] q2m2;
				q2m1 -> [edge, bend left = 50, dashed] q2m2;
				
				q2m2 -> [edge, loop above] q2m2;
			};
			\end{tikzpicture}
			\begin{tikzpicture}[->,>=latex',
			vertex/.style={circle, draw=black, fill, minimum width=1.5mm, inner sep=0pt, outer sep=0pt},
			every label/.style={inner sep=0pt, minimum width=0pt, label distance=0.1mm},
			yscale=1.5,
			xscale=1.75
			]
			\graph[nodes=vertex, empty nodes, no placement] {
				
				q2m2[x=-2,y=0,label=below:$q_{2m + 2}$] -> [edge, loop above] q2m2;
				
				q4m3[x=5,y=0,label=below:$q_{4m + 3}$];
				
				q4m2[x=4,y=0,label=below:$q_{4m + 2}$] -> 
				q4m1[x=3,y=0,label=below:$q_{4m + 1}$] -> [edge label=$\ldots$]
				q3m3[x=2,y=0,label=below:$q_{3m + 3}$] ->
				q3m2[x=1,y=0,label=below:$q_{3m + 2}$] -> [edge label=$\cdots$]
				q2m4[x=0,y=0,label=below:$q_{2m + 4}$] ->
				q2m3[x=-1,y=0,label=below:$q_{2m + 3}$] ->  [edge label=$\ldots$, swap] q2m2;
				
				q2m3 -> [edge, bend left = 50, dashed] q3m2;
				q2m4 -> [edge, bend left = 50, dashed] q3m3;
				q3m2 -> [edge, bend left = 50, dashed] q4m1;
				q3m3 -> [edge, bend left = 50, dashed] q4m2;
				q4m1 -> [edge, bend left = 50, dashed] q4m3;
				q4m2 -> [edge, bend left = 50, dashed] q4m3;
				
				q4m3 -> [edge, loop above] q4m3;
			};
			\end{tikzpicture}
			\caption{The automaton providing a lower bound for subset synchronization in binary monotonic automata. Dashed arrows represent transitions for the letter $0$, solid -- for the letter $1$. The states $q_1, q_{2m + 2}, q_{4m + 3}$ are sink states. The picture is divided into two parts because of its width.} \label{fig-subset-2}
		\end{center}
	\end{figure}
	
	Define $S = \{q_{m + 2}, q_{4m + 2}\}$. Let us prove that a shortest word synchronizing $S$ has length at least $m^2$. 
	
	The set $S$ can only be mapped to $q_{2m + 2}$, since it is a sink state between $\min S$ and $\max S$. Thus, no word synchronizing $S$ maps any its state to $q_1$ or $q_{4m + 3}$. Consider now an interval $[q_i, q_j]$ for $2 \le i \le m + 1$, $2m + 3 \le j \le 4m + 2$ and note that applying $0$ reduces its length by $1$ (or maps its right end to $q_{4m + 3}$), and applying $1$ maps its ends to the ends of another interval of this form with the same length (or maps its left end to $q_1$). The maximal length of a segment of this form that allows its left end to be mapped to $q_{2m + 2}$ is $2m + 1$, so before any end of the interval is mapped to $q_{2m + 2}$, the letter $0$ has to be applied at least $m$ times. Each application of $0$ moves the right end of the intervals $m - 1$ states to the right, so each application of $0$ requires $m - 1$ applications of $1$ so that $0$ can be applied one more time. Thus, the word mapping $S$ to $q_{2m + 2}$ has length at least $m^2$. Note that $S$ can be synchronized by a word $w = (1^{m - 1} 0)^m 1^{2m}$ of length $|w| = m^2 + 2m$.
\end{proof}

For a $n$-state binary monotonic automaton we get a lower bound of $\frac{(n - 3)^2}{16}$ from this theorem.

By removing the first and the last state (and leaving all transitions to the removed states undefined) in the automata in the both series, we get the following results.

\begin{corollary}
	For infinitely many $n$, there exists a $n$-state monotonic partial automaton over a three-letter alphabet with shortest carefully synchronizing word of length at least $\frac{(n-1)^2}{4} + 1$.
\end{corollary}

\begin{corollary}	
	For infinitely many $n$, there exist a $n$-state monotonic binary partial automaton with shortest carefully synchronizing word of length at least $\frac{(n - 1)^2}{16}$.
\end{corollary}

\begin{question}
	Find upper bounds on the length of a shortest carefully synchronizing words for monotonic partial automata.
\end{question}

A related question is to measure the shortest length of a word accepted simultaneously by $k$ monotonic automata. This question is important, for example, for investigation of the computational complexity of the {\sc Finite Automata Intersection} problem, see Section \ref{sec-complex} for the details. For alphabet of unbounded size, a partial answer is provided by the following theorem.

\begin{theorem} \label{tm-acc-k}
	For any $k > 0$, there exist $k$ $3$-state monotonic automata over a $k$-letter alphabet, such that the length of a shortest word accepted by all automata is at least $2^k - 1$.
\end{theorem}
\begin{proof}
	To prove the theorem, we show how to imitate a simple binary counter with $k$ monotonic automata. Consider the following family $A_i = (Q_i, \Sigma, \delta_i)$, $1 \le i \le k$. Here $Q_i = \{f_i, s_i, t_i\}$, $\Sigma = \{a_1, \ldots, a_k\}$ and $\delta_i$ are defined as follows. We define $\delta_i(s_i, a_i) = t_i$, and $\delta_i(s_i, a_j) = f_i$,  $\delta_i(t_i, a_j) = s_i$ for $i < j$. All yet undefined transitions are self-loops. In each $A_i$, we take $s_i$ as the initial state, and $t_i$ as the only accepting state.
	
	By induction, the length of a shortest word accepted by all automata is $2^{k} - 1$, since each next automaton can be mapped to its accepting state only when all previous automata are already in their accepting states, and this operation maps all previous automata to their initinal states.
\end{proof}

Hence, we get a lower bound of $2^{\frac{n}{3}} - 1$, where $n$ is the total number of states in the automata. Thus the most obvious candidate for a certificate to show that {\sc Finite Automata Intersection} is in NP for monotonic automata fails.

The proof of Theorem \ref{tm-acc-k} implies the following interesting result.

\begin{corollary}
	For each $k$ there exists a monotonic partial automaton with $3k$ states, alphabet of size $k + 1$ and a shortest carefully synchronizing word of length at least $3^{k}$.
\end{corollary}
\begin{proof}
	Remove the states $f_i$ in all automata from the construction of Theorem \ref{tm-acc-k} and leave all transtions to this states undefined. Then add a letter $a$ such that $\delta(t_i, a) = t_k$ for $1 \le i \le k$. Thus we get a carefully synchronizing automaton imitating a binary counter. The last step is to note that $2^{\frac{1}{2}} < 3^\frac{1}{3}$, and to imitate a ternary counter in the exactly same way instead.
\end{proof}
	
We note that the ternary counter is optimal since $2^{\frac{1}{2}} < 3^\frac{1}{3}$ and $3^\frac{1}{3} > \ell^\frac{1}{\ell}$ for $\ell \ge 4$.

The bound in this corollary is of the same order as in the result of Rystsov \cite{Rystsov1980} and Martyugin \cite{Martyugin2010}, but our example is monotonic and has simpler structure. Since each carefully synchronizing partial automaton  has a carefully synchronizing word of length $O(3^{\frac{n}{3}})$ \cite{Rystsov1980}, the bound is asymptotically tight.

Some additional optimization can be done for the described construction. In particular, the after countng to $3^k - 1$, only two states in the set of reached states can be synchronized by a new letter. Then counting is repeated (to $3^{k - 1} - 1$) and again two states are synchronized, and so on. Thus, the bound will become $3^k + 3^{k - 1} + \ldots + 3 + 1$. However, the bound remains of the same order and is still smaller than the bounds of Rystsov and Martyugin for general automata.

For a constant number of letters the considered problems remain open.

\begin{question}
	What is the length of a shortest word accepted simultaneously by $k$ monotonic automata over an alphabet of constant size? What is the length of a shortest word carefully synchronizing a monotonic partial automaton with $n$ states and alphabet of constant size?
\end{question}

\section{Complexity Results}\label{sec-complex}

In this section, we obtain computational complexity results for several problems related to subset synchronization in monotonic automata. We improve Eppstein's construction \cite{Eppstein1990} to make it suitable for monotonic automata. We shall need the following NP-complete {\sc SAT} problem \cite{Sipser2006}.

\begin{tabular}{||p{36em}}
	~{\sc SAT} \\
	~{\em Input}: A set $X$ of $n$ boolean variables and a set $C$ of $m$ clauses;\\
	~{\em Output}: Yes if there exists an assignment of values to the variables in $X$ such that all clauses in $C$ are satisfied, No otherwise.
\end{tabular}

\vspace{.1cm}

Provided a set $X$ of boolean variables $x_1, \ldots, x_n$ and a clause $c_j$, construct the following automaton $A_j = (Q, \Sigma, \delta)$. Take 

$$Q = \{q_1, \ldots, q_{n + 1}\} \cup \{q'_2, \ldots, q'_n\} \cup \{s, t\}.$$

Let $\Sigma = \{0, 1, r\}$. Define the transition function $\delta$ as follows. For each $i$, $1 \le i \le n$, map a state $q_i$ to $q'_{i + 1}$ (or to $t$ if $i = n$) by a letter $x \in \{0, 1\}$ if the assignment $x_i = x$ satisfies $c_j$, and to $q_{i + 1}$ otherwise. For each $i$, $2 \le i \le n - 1$, set $\delta(q'_i, x) = \delta(q'_{i + 1}, x)$ for $x \in \{0, 1\}$. Set $\delta(q'_n, x) = t$ for $x \in \{0, 1\}$. Define transitions from $t$ for letters $0, 1$ as self-loops. Finally, define $\delta(q, r) = s$ for $q \in Q \setminus \{t\}$, $\delta(t, r) = t$. See Figure \ref{fig-clause} for an example.

\setlength{\unitlength}{2pt}
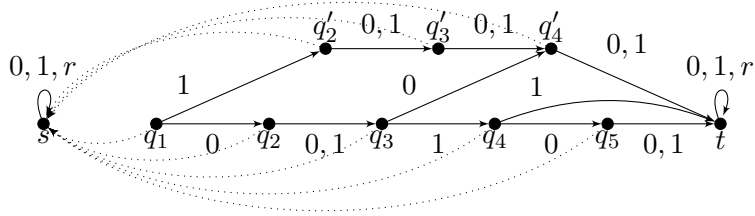
\begin{figure}[hbt]
	\begin{center}
		\begin{tikzpicture}[->,>=latex',
		vertex/.style={circle, draw=black, fill, minimum width=1.5mm, inner sep=0pt, outer sep=0pt},
		every label/.style={inner sep=0pt, minimum width=0pt, label distance=0.1mm},
		yscale=1,
		xscale=1.5
		]
		\graph[nodes=vertex, empty nodes, no placement] {
			
			s[x=0,y=0,label=below:$s$];
			
			q1[x=1,y=0,label=below:$q_1$] -> [edge label={$0$}, swap]
			q2[x=2,y=0,label=below:$q_2$] -> [edge label={$0, 1$}, swap]
			q3[x=3,y=0,label=below:$q_3$] -> [edge label={$1$}, swap]
			q4[x=4,y=0,label=below:$q_4$] -> [edge label={$0$}, swap]
			q5[x=5,y=0,label=below:$q_5$] -> [edge label={$0, 1$}, swap]
			
			t[x=6,y=0,label=below:$t$];
			
			qp2[x=2.5,y=1,label=above:$q'_2$] -> [edge label={$0, 1$}]
			qp3[x=3.5,y=1,label=above:$q'_3$] -> [edge label={$0, 1$}]
			qp4[x=4.5,y=1,label=above:$q'_4$] -> [edge label={$0, 1$}, near start]
			t;
			
			qp2 -> [edge, bend right = 50, dotted] s;
			qp3 -> [edge, bend right = 50, dotted] s;
			qp4 -> [edge, bend right = 50, dotted] s;
			
			q1 -> [edge label={$1$}, near start] qp2;
			q3 -> [edge label={$0$}, near start] qp4;
			q4 -> [edge label = {$1$}, bend left = 30, near start] t;
			
			q1 -> [edge, bend left = 50, dotted] s;
			q2 -> [edge, bend left = 50, dotted] s;
			q3 -> [edge, bend left = 50, dotted] s;
			q4 -> [edge, bend left = 50, dotted] s;
			q5 -> [edge, bend left = 50, dotted] s;
			
			t -> [edge label={$0, 1, r$}, loop above] t;
			
			s -> [edge label={$0, 1, r$}, loop above] s;
		};
		\end{tikzpicture}
		\caption{The automaton $A_j$ for a clause $c_j = (x_1 \vee \overline{x}_3 \vee x_4)$. Dotted arrows represent transitions for the letter $r$.} \label{fig-clause}
	\end{center}
\end{figure}

Note that $A_j$ is monotonic, since it respects the order

$$s, q_1, q_2, q'_2, q_3, q'_3, \ldots, q_n, q'_n, q_{n + 1}, t.$$

It is also weakly acyclic, since its underlying digraph has no simple cycles of length at least $2$.

Also, provided the number of variables $n$, construct an automaton $T = (Q_T, \Sigma, \delta_T)$ as follows. Take $Q_T = \{a, p_1, \ldots, p_{n + 1}, b\}$, $\Sigma = \{0, 1, r\}$. Define $\delta(p_i, x) = p_{i + 1}$ for each $i$, $1 \le i \le n$, and $x \in \{0, 1\}$, and $\delta(p_{n + 1}, x) = b$ for $x \in \{0, 1\}$. Define also $\delta(a, x) = a$ and $\delta(b, x) = b$ for each $x \in \Sigma$, and $\delta(p_i, r) = a$ for $1 \le i \le n + 1$. See Figure \ref{fig-timer} for an example. This automaton is monotonic, since it respects the order $a, p_1, \ldots, p_{n + 1}, b$, and it is obviously weakly acyclic.

\setlength{\unitlength}{2pt}
\begin{figure}[hbt]
	\begin{center}
		\begin{tikzpicture}[->,>=latex',
		vertex/.style={circle, draw=black, fill, minimum width=1.5mm, inner sep=0pt, outer sep=0pt},
		every label/.style={inner sep=0pt, minimum width=0pt, label distance=0.1mm},
		yscale=1,
		xscale=1.5
		]
		\graph[nodes=vertex, empty nodes, no placement] {
			
			a[x=0,y=0,label=below:$a$];

			p1[x=1,y=0,label=below:$p_1$] -> [edge label={$0,1$}]
			p2[x=2,y=0,label=below:$p_2$] -> [edge label={$0,1$}]
			p3[x=3,y=0,label=below:$p_3$] -> [edge label={$0,1$}]
			p4[x=4,y=0,label=below:$p_4$] -> [edge label={$0,1$}]
			p5[x=5,y=0,label=below:$p_5$] -> [edge label={$0,1$}]
			
			b[x=6,y=0,label=below:$b$];

			p1 -> [edge, bend left = 50, dotted] s;
			p2 -> [edge, bend left = 50, dotted] s;
			p3 -> [edge, bend left = 50, dotted] s;
			p4 -> [edge, bend left = 50, dotted] s;
			p5 -> [edge, bend left = 50, dotted] s;
			
			a -> [edge label={$0, 1, r$}, loop above] a;
			
			b -> [edge label={$0, 1, r$}, loop above] b;
		};
		\end{tikzpicture}
		\caption{The automaton $T$ for $n = 4$ variables. Dotted arrows represent transitions for the letter $r$.} \label{fig-timer}
	\end{center}
\end{figure}
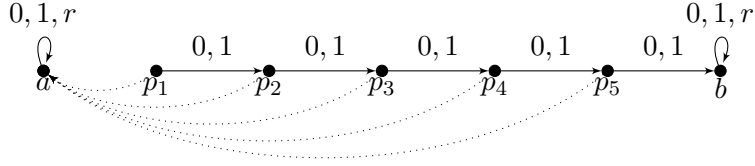

First, we prove NP-completeness of the following problem.

\begin{tabular}{||p{36em}}
	~{\sc Finite Automata Intersection} \\
	~{\em Input}: Automata $A_1, \ldots, A_k$ (with initial and accepting states); \\
	~{\em Output}: Yes if there is a word which is accepted by all automata, No otherwise.
\end{tabular}

\vspace{.1cm}

This problem is PSPACE-complete for general automata \cite{Kozen1977}, and NP-complete for binary weakly acyclic automata \cite{Ryzhikov2017}. Some results on this problem are surveyed in \cite{Holzer2011}. Blondin et al. \cite{Blondin2016} provide further results on the problem. Some complexity lower bounds are presented in \cite{Wehar2014} and \cite{Fernau2016}.

\begin{theorem} \label{tm-inters}
	The {\sc Finite Automata Intersection} problem is NP-complete for monotonic weakly acyclic automata over a three-letter alphabet.
\end{theorem}
\begin{proof}
	The fact that the problem is in NP follows from the fact that {\sc Finite Automata Intersection} for weakly acyclic automata is in NP \cite{Ryzhikov2017}.
	
	To prove hardness, we reduce the {\sc SAT} problem. For each clause $c_j \in C$, construct an automaton $A_j$, and set $q_1$ as its initial state and $t$ as its only accepting state. Construct also the automaton $T$ with the initial state $p_1$ and accepting state $a$.
	
	We claim that $C$ is satisfiable if and only if all automata in $\{A_j \mid c_j \in C\} \cup \{T\}$ accept a common word $w$. Indeed, assume that there is a common word accepted by all these automata. Then none of the first $n$ letters of this word can be $r$, otherwise all automata $A_j$ are mapped to $s$, which is a non-accepting sink state. The next letter has to be $r$, otherwise $T$ is mapped to $b$, which is a non-accepting sink state. But that means that in each $A_j$, the set $q_1$ is mapped by a $n$-letter word $z_1 \ldots z_n$ to the accepting state $t$. Thus, by construction, the assignment $x_i = z_i$ satisfies all clauses in $C$.
	
	By the same reasoning, if the assignment $x_i = z_i$, $1 \le i \le n$, satisfies all clauses in $C$, then $z_1 \ldots z_n r$ is a word accepted by all automata. 
\end{proof}

Now we switch to a related {\sc Set Rank} problem.
 
\begin{tabular}{||p{30em}}
	~{\sc Set Rank} \\
	~{\em Input}: An automaton $A$ and a set $S$ of its states;\\
	~{\em Output}: The rank of $S$. 
\end{tabular}

\vspace{.1cm}

This problem is hard to approximate for binary weakly acyclic automata \cite{Ryzhikov2017}. To get inapproximability results for monotonic automata, we use the following problem.


\begin{tabular}{||p{36em}}
	~{\sc Max-$3$SAT} \\
	~{\em Input}: A set $X$ of $n$ boolean variables and a set $C$ of $m$  $3$-term clauses;\\
	~{\em Output}: The maximum number of clauses that can be simultaneously satisfied by some assignment of values to the variables.
\end{tabular}

\vspace{.1cm}

This problem cannot be approximated in polynomial time within a factor of $\frac{7}{8} + \epsilon$ for any $\epsilon > 0$ unless P = NP \cite{Hastad2001}.

\begin{theorem}\label{tm-rank}
	The {\sc Set Rank} problem cannot be approximated in polynomial time within a factor of $\frac{9}{8} - \epsilon$ for any $\epsilon > 0$ in monotonic weakly acyclic automata over a three-letter alphabet unless P~=~NP.
\end{theorem}
\begin{proof}
		We reduce the {\sc Max-$3$SAT} problem. For each clause $c_j \in C$, construct an automaton $A_j$. Construct also $m$ copies of the automaton $T$, denoted $T_j$, $1 \le j \le m$. Define an automaton $A$ with the set of states which is the union of all sets of states of $\{A_j, T_j \mid 1 \le j \le m\}$, alphabet $\Sigma$ and transition functions defined in all constructed automata. For each $j$, identify the state $t$ in $A_j$ with the state $a$ in $T_j$. Take $S$ to be the set of states $q_1$ from each automaton $A_j$, and $p_1$ from each $T_j$. The constructed automaton is monotonic and weakly acyclic.

		If $h$ is the minimum number of clauses in $C$ that are not satisfied by an assignment, the set $S$ has rank $m + h$. Indeed, consider an assignment $x_i = z_i$, $1 \le i \le n$, not satisfying exactly $h$ clauses in $C$. Then the word $z_1 \ldots z_n r$ has rank $m + h$ with respect to $S$.
		
		In the other direction, let $w$ be a word of minimum rank with respect to the set $S$. If any of the first $n$ letters of $w$ is $r$, then $q_1$ in each $A_i$ is mapped to $s$ in the corresponding automaton, and thus $w$ has rank $2m$ with respect to $S$. The same is true if $(n + 1)$st letter of $w$ is not $r$, because then $p_1$ in each $T_i$ is mapped to $b$ in the corresponding automaton. If first $n$ letters $z_1, \ldots, z_n$ of $w$ are not $r$, and the next letter is $r$, then the assignment $x_i = z_i$ does not satisfy exactly $h'$ clauses, where $m + h'$ is the rank of the word $w$ with respect to $S$. For the word of minimum rank, we get the required equality.
		
		It is NP-hard to decide between (i) all clauses in $C$ are satisfiable and (ii) at most $(\frac{7}{8} + \epsilon)m$ clauses in $C$ can be satisfied by an assignment \cite{Hastad2001}. In the case (i), the rank of $S$ is $m$, in the case (ii) it is at least $m + (\frac{1}{8} - \epsilon)m$. Since it is NP-hard to decide between this two options, we get $(\frac{9}{8} - \epsilon)$-inapproximability for any $\epsilon > 0$.
\end{proof}

By using an argument similar to the proof of Theorem \ref{tm-rank}, we can show inapproximability of the maximization version of {\sc Finite Automata Intersection} (where we are asked to find a maximum number of automata accepting a common word). Indeed, take $m$ copies of $T$ together with the set $\{A_j \mid c_j \in C\}$ as the input of {\sc Finite Automata Intersection} and reduce {\sc Max-$3$SAT} to it (input and accepting states are assigned according to the construction in Theorem \ref{tm-inters}). Then the maximum number of automata accepting a common word is $m + g$, where $g$ is the maximum number of simultaneously satisfied clauses in $C$, since all copies of $T$ have to accept this word. Thus it is NP-hard to decide between (i) all $2m$ automata accept a common word and (ii) at most $m + (\frac{7}{8} + \epsilon)m$ automata accept a common word, and we get the following result.

\begin{corollary}
	The maximization version of the {\sc Finite Automata Intersection} problem cannot be approximated in polynomial time within a factor of $\frac{15}{16} + \epsilon$ for any $\epsilon > 0$ in monotonic weakly acyclic automata over a three-letter alphabet unless P~=~NP.
\end{corollary}

Consider now the following problem briefly discussed in the introduction.

\begin{tabular}{||p{30em}}
	~{\sc Careful Synchronization} \\
	~{\em Input}: A partial automaton $A$; \\
	~{\em Output}: Yes if $A$ is carefully synchronizing, No otherwise.
\end{tabular}

\begin{theorem}
	The {\sc Careful Synchronization} problem is NP-hard for monotonic  automata over a four-letter alphabet.
\end{theorem}
\begin{proof}
	We reduce the {\sc SAT} problem. Let $A'_j$ be $A_j$ with alphabet restricted to $\{0, 1\}$ and without the state $s$, and $T'$ be $T$ with alphabet restricted to $\{0, 1\}$ and without the states $a, b$. Let $\Sigma' = \{0, 1, y, z\}$. Provided $X$ and $C$, construct $A'_j$ for each clause $c_j$, and also construct $T'$. Let $Q'$ be the union of all states of each $A'_j$, all states of $T'$ and a new state $f$. We expand already defined (for the letters $0, 1$) transition function $\delta$ as follows. Define $y$ to map all states in each $A'_j$ to $q_1$ in this gadget, and all states in $T'$ to $p_1$. Finally, define $z$ to map the state $t$ in each $A_j$, the state $p_{n + 1}$ in $T$ and the state $f$ to $f$. Leave all other transition undefined. We denote thus obtained automaton as $A' = (Q', \Sigma', \delta')$.
	
	Any word $w$ carefully synchronizing $A'$ begins with $y$, since it is the only letter defined for all states. After applying it, the set $S$ of reached states consists of $q_1$ in each $A_j$, $p_1$ in $T$, and $f$. To reach $f$ from this set, we need to apply $z$ at least once. Right before applying $z$, the set of reached states must consist of exactly $n$ letters of the set $\{0, 1\}$, because application of $y$ takes us back to $S$, and applying $0, 1$ more than $n$ times is not defined for the state $p_1$. Thus, in each $A_j$ the state $q_1$ must be mapped to $t$ by $n$ $0$s and $1$s which is possible if and only if $C$ is satisfiable.
\end{proof}

\begin{question}
	What is the complexity of the mentioned problems for binary monotonic automata? Do the problems discussed in this section belong to NP for monotonic automata?
\end{question}

We note that it does not matter in the provided reductions whether a linear order preserved by all transitions is known or not.

\section{Subset Road Coloring} \label{sec-road}

The famous Road Coloring problem is formulated as follows. Given a strongly connected digraph with all vertices of equal out-degree $k$, is it possible to find a coloring of its arcs with letters of alphabet $\Sigma$, $|\Sigma| = k$, resulting in a synchronizing deterministic automaton. This problem was stated in 1977 by Adler, Goodwyn and Weiss \cite{Adler1977} and solved in 2007 by Trahtman \cite{Trahtman2009}. A natural generalization of this problem is to find a coloring of a strongly connected digraph turning it into a deterministic automaton where a given subset of states is synchronizing. We introduce the problem formally and show that its solution is a corollary of a result of B\'{e}al and Perrin \cite{Beal2014}. In particular, the problem of deciding whether such a coloring exists is solvable in polynomial time.

Let $G = (V, E)$ be a strongly connected digraph such that each its vertex has out-degree $k$. A {\em coloring} of $G$ with letters from alphabet $\Sigma$, $|\Sigma| = k$, is a function assigning each arc of $G$ a letter from $\Sigma$, such that for each vertex, each pair of arcs outgoing from it achieves different letters. We say that a coloring {\em synchronizes} $S \subseteq V$ in $G$ if $S$ is a synchronizing set in the resulting automaton.

If the greatest common divisor of the lengths of all cycles of $G$ is $\ell$, the set $V$ can be partitioned into sets $V_1, \ldots, V_\ell$ in such a way that if $(v, u)$ is an arc of $G$, then $v \in V_i, u \in V_{i + 1}$ or $v \in V_\ell, u \in V_1$ \cite{Friedman1990}. Moreover, such partition is unique.

\begin{theorem}
	A strongly connected digraph $G$ with vertices of equal out-degree has a coloring synchronizing a set $S \subseteq V$ if and only if $S \subseteq V_i$ for some $i$.
\end{theorem}
\begin{proof} Obviously, if two vertices of $S$ belong to distinct sets $V_i$ and $V_j$, $S$ can not be synchronized. Assume that $S \subseteq V_i$ for some $i$. As proved in \cite{Beal2014}, there exists a coloring of $G$ such that the resulting automaton $A$ has rank $\ell$. In this coloring each $V_j$, $1 \le j \le \ell$, is a synchronizing set, since no two states from two different sets $V_p, V_t$, $p \not = t$, can be synchronized and $A$ has rank $\ell$. Hence, $S \subseteq V_i$ is also a synchronizing set.
\end{proof}

According to this theorem, checking whether there exists such a coloring can be performed in polynomial time. This coloring can be constructed in polynomial time using the algorithm from \cite{Beal2014}.

\subsubsection*{Acknowledgments}

We thank Ilia Fridman for useful comments on the presentation of this paper, and anonymous reviewers for multiple useful comments and suggestions that improved the content and presentation of the preliminary version of this paper.

{\footnotesize
\bibliography{SyncBib}
	\bibliographystyle{alpha}}
	
\end{document}